\documentclass[a4paper,12pt]{article}
\usepackage{amsmath,amsthm}
\usepackage{amsfonts}
\usepackage{graphicx}
\usepackage[utf8]{inputenc}
\usepackage{mathtools}
\usepackage{xcolor}
\usepackage{amssymb}
\usepackage{geometry}
\usepackage{hyperref}
\usepackage{ulem}
\usepackage{comment}
\usepackage{authblk}

\newcommand{\blue}[1]{\textcolor{blue}{#1}}

\newtheorem{theorem}{Theorem}
\newtheorem{remark}{Remark}

\DeclareMathOperator{\sign}{sign}
\title{Purcell swimmer near a wall}
\date{ }
\author{Enrico Micalizio$^1$, Marco Morandotti$^1$, Henry Shum$^2$, Marta Zoppello}
\affil[1]{Department of Mathematical Sciences, Politecnico di Torino, Corso Duca degli Abruzzi, 24, Torino, 10129,Italy}
\affil[2]{Department of Applied Mathematics, University of Waterloo, Waterloo, ON N2L 3G1, Canada}

\begin{document}

\maketitle

\begin{abstract}
We study the effects of hydrodynamic interactions between a wall and the Purcell three-link swimmer in the two-dimensional case. 
After deriving the equations of motion in a low Reynolds number regime using Resistive Force Theory with suitably modified drag coefficients, we show, by means of criteria from Geometric Control Theory, that the system is controllable at configurations that are nearly parallel to the wall. 
Furthermore, we study configurations that are tilted, and we show net displacement with respect to the initial orientation.
Some numerical experiments illustrate the analytical results.
\end{abstract}


\section{Introduction}
Low-Reynolds-number locomotion has been the subject of sustained interest since Purcell's seminal work on life at low Reynolds numbers \cite{purcell}. In this regime, inertia is negligible and the motion of a swimmer results from non-reciprocal shape changes interacting with viscous forces. Among the paradigmatic models introduced by Purcell, the three-link swimmer now commonly referred to as the Purcell swimmer has become a key model for the theoretical study of microswimming, serving as a minimal yet rich system for the analysis of controllability, motion planning, and geometric aspects of locomotion.

In an unbounded planar fluid, the controllability properties of the Purcell swimmer are now well understood \cite{zop13,BECKER_KOEHLER_STONE_2003}. Modeling the swimmer as a kinematic control system driven by internal shape variables, several authors have shown that it is locally controllable around straight or moderately bent configurations \cite{zop15,WiezelOr2016}. These results rely on Lie algebraic techniques and exploit the structure induced by the Stokes equations, which yield a driftless system with control vector fields associated with joint actuation.

By contrast, the effect of boundaries on controllability has received comparatively less attention, despite its clear physical relevance. In realistic environments, microswimmers often operate in confined or partially bounded domains, where nearby walls modify hydrodynamic interactions through no-slip boundary conditions. From a control-theoretic viewpoint, the presence of a wall could alter the resistance matrix and breaks some of the symmetries present in the unbounded case. This naturally raises the question of whether such a boundary hinders controllability by restricting the accessible motions, or whether it may instead enrich the swimmer's dynamics and potentially enhance maneuverability.

The first objective of this paper is to address this question for a Purcell swimmer swimming in the plane near a single planar wall. Focusing on a configuration in which the swimmer is aligned and parallel to the wall, we derive the associated equation of motion using resistive force theory with wall-induced hydrodynamic corrections \cite{Katz_Blake_Paveri-Fontana_1975,LAUGA2006400}, casting it as a control system. Our main theoretical result (Theorem~\ref{main_thm}) shows that the presence of the wall is not an obstruction to controllability: the swimmer remains locally controllable near this aligned configuration. This demonstrates that, at least in this regime, confinement does not destroy the swimmer's ability to generate arbitrary small motions.
Beyond the analysis of controllability via Lie bracket approximations, we investigate the net displacement generated by a classical small-amplitude Purcell stroke when the swimmer starts from a straight configuration tilted with respect to the wall. By explicitly computing the induced motion, we show that the stroke produces a systematic drift in the direction of the initial orientation, without generating any net rotation. We further establish that the presence of the wall affects the norm of the displacement which is no longer constant as in an unbounded fluid domain; this norm is maximized when the swimmer is parallel to the wall.

This behavior differs from that reported in numerical simulations and experiments on swimming spermatozoa and bacteria, where more detailed descriptions of swimming gaits and hydrodynamic interactions typically predict a reorientation and progressive alignment with the wall, often leading to surface accumulation \cite{berke_hydrodynamic_2008, frymier_three-dimensional_1995, rothschild_non-random_1963} or, in some regimes, escape \cite{shum_modelling_2010,smith_human_2009}.
Although the forward-swimming gait for the Purcell swimmer is not hydrodynamically deflected by the wall as it approaches, we show with results from a control-oriented framework based on Lie bracket computations that the swimmer can be steered both toward and away from the wall by adjusting its gait. 
\section{Mathematical model}\label{model2D}
We consider a system composed of a three-link swimmer and a rigid wall, modeled as a straight line in our 2D model, and we adopt a reference frame in which the wall coincides with the $x$-axis. 
Let us denote by $y_t$ the distance, at time $t$, between the wall and the midpoint of the central link, whose position in the plane, at time $t$, is $\mathbf{x}_t=(x_t\,,y_t)$. 
Moreover, the orientation of the central link with respect to the wall is described by the angle $\theta_t \in [-\pi, \pi)$, while the relative angles between the central link and the two outer links are denoted by $\alpha_t^{(\pm1)}\in [-\pi,\pi)$.

    \begin{figure}
        \centering
        \includegraphics{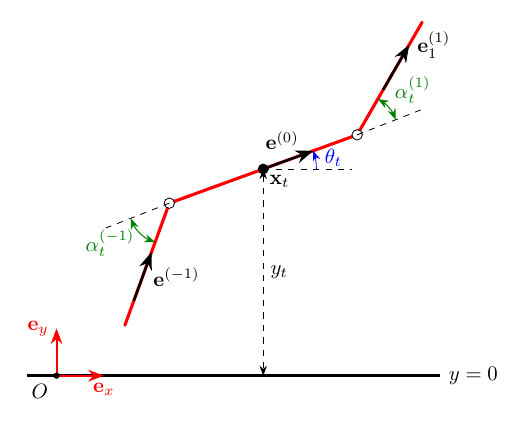} 
        \caption{Representation of the Purcell swimmer near the wall.}
    \end{figure}

Thus, the triple $(x_t\,, y_t\,, \theta_t)$ describes the global position of the swimmer at time $t$, whereas the pair $\boldsymbol{\alpha}_t\coloneqq\big[\alpha_t^{(-1)}, \alpha_t^{(1)}\big]^\top$ represents the shape.

Each link has constant length $L_i$ ($i\in\{\pm1,0\}$) and thickness $r$, and we assume that 
\[
0 < r \ll L_i\quad \text{for every $i\in\{\pm1,0\}$}
\qquad\text{and}\qquad
0 < r \ll y_t,\quad \text{for every $t$.}
\]
We denote by $\mathbf{e}_t^{(i)}$ the unit vector describing the direction of the $i$-th link, 
defined as
\[
\mathbf{e}_t^{(i)} \coloneqq 
\begin{bmatrix}
    \cos\big(\theta_t + i^2\alpha_t^{(i)}\big)\\
    \sin\big(\theta_t + i^2\alpha_t^{(i)}\big)
\end{bmatrix}, \qquad \text{for $i=-1,0,1$} 
\]
(we can fictitiously define $\alpha_t^{(0)}=0$; its contribution is irrelevant).
Finally, we can express the 
position of a point $\mathbf{x}_t^{(i)}(s)$ on the $i$-th link 
\begin{equation}\label{62}
\mathbf{x}_t^{(i)}(s) = \mathbf{x}_t + \sigma_-^{(i)}
\frac{L_0}{2}\mathbf{e}_t^{(0)} + \sigma_+^{(i)} 
s\,\mathbf{e}_t^{(i)},
\end{equation}
where $s \in [0, L_i]$ denotes the distance along the $i$-th link from its hinge, and where we have defined
$\sigma_\pm^{(i)}\coloneqq \sign(i\pm\frac{1}{2})$.\footnote{The explicit values of $\sigma_\pm^{(i)}$ are $\sigma_-^{(-1)}=\sigma_-^{(0)}=\sigma_+^{(-1)}=-1$ and $\sigma_-^{(1)}=\sigma_+^{(0)}=\sigma_+^{(1)}=1$. 
}
Taking the time derivative in \eqref{62}, the velocity of the generic point in the three-link swimmer reads
\begin{equation}\label{83}
\dot{\mathbf{x}}_t^{(i)}(s) =  \dot{\mathbf{x}}_t +\sigma_-^{(i)} 
\frac{L_0}{2}\dot{\theta}_t\,\mathbf{n}_t^{(0)} + \sigma_+^{(i)} 
s\,\big(\dot{\theta}_t + i^2\dot{\alpha}_t^{(i)}\big)\mathbf{n}_t^{(i)},
\end{equation}
where $\mathbf{n}_t^{(i)} \coloneqq R_{\pi/2}\mathbf{e}_t^{(i)}$
is the normal vector obtained 
by rotating the tangent vector~$\mathbf{e}_t^{(i)}$ by $\pi/2$.

Now, we aim to apply Resistive Force Theory \cite{gray_propulsion_1955} to compute the hydrodynamic force and torque densities acting on the swimmer. However, in the presence of a wall, the drag coefficients (usually denoted by $C_{\parallel}$ and $C_{\perp}$ when they are constant) depend explicitly on the distance $d_t^{(i)}(s)=y_t^{(i)}(s)$ of the point $\mathbf{x}_t^{(i)}(s)$ from the wall. 
According to \cite{Katz_Blake_Paveri-Fontana_1975}, the explicit form for these coefficients is the following

\begin{equation} \label{coefficients}
   C_{\perp t}^{(i)}(s) = 4 \pi \mu \log^{-1} \bigg(\frac{2\,d_t^{(i)}(s)}{r}\bigg) =2C_{\parallel t}^{(i)}(s). \end{equation}

In the sequel, we will omit the explicit dependence of the drag coefficients from~$t$.
Let $\mathbf{f}^{(i)}_t(s)$ be the hydrodynamic force density acting on the point $\mathbf{x}_t^{(i)}(s)$ of the $i$-th link. 
Resistive Force Theory yields
\begin{equation} \label{forcedensity}
\begin{split}
\!\! \mathbf{f}_t^{(i)}(s) &= \big[(C_{\parallel}^{(i)}(s) - C_{\perp}^{(i)}(s))\mathbf{e}_t^{(i)} \otimes \mathbf{e}_t^{(i)} + C_{\perp}^{(i)}(s)\mathbf{I}\big]\dot{\mathbf{x}}_t^{(i)} (s) \\
&= C_\perp^{(i)}(s) \bigg[\bigg(\mathbf{I}_2-\frac{1}{2}\mathbf{E}_t^{(i)}\bigg)\dot{\mathbf{x}}_t+\sigma_+^{(i)}\,s\big(\dot\theta_t+i^2\dot\alpha_t^{(i)}\big) \mathbf{n}_t^{(i)} \\
&\phantom{C_\perp^{(i)}(s) \bigg[}\quad +\sigma_-^{(i)}\frac{L_0}{2}\dot\theta_t \bigg(\mathbf{I}_2-\frac{1}{2}\mathbf{E}_t^{(i)}\bigg)\mathbf{n}_t^{(0)}\bigg],
\end{split}
\end{equation}
where we defined $\mathbf{E}_t^{(i)} \coloneqq \mathbf{e}^{(i)}_t \otimes \mathbf{e}^{(i)}_t$ for notational convenience, and $\mathbf{I}_2$ is the $2\times2$ identity matrix.
Using \eqref{forcedensity} and defining $\sigma_{*}^{(i)}\coloneqq\sign(i^2-\frac12)$,\footnote{The explicit values of $\sigma_*^{(i)}$ are $\sigma_*^{(\pm1)}=1$ and $\sigma_*^{(0)}=-1$. Notice that $\sigma_*^{(i)}=\sigma_-^{(i)}\sigma_+^{(i)}$.} the torque density about the reference point $\mathbf{x}_t$ is given by
\begin{equation} \label{torquedensity}
\begin{split}
\tau_t^{(i)}(s) &= C_\perp^{(i)}(s)\bigg[\sigma_-^{(i)}\frac{L_0}{2}\mathbf{n}_t^{(0)}\cdot\bigg(\mathbf{I}_2-\frac{1}{2}\mathbf{E}_t^{(i)}\bigg)\dot{\mathbf{x}}_t 
    + \sigma_+^{(i)}s\, \mathbf{n}_t^{(i)}\cdot\bigg(\mathbf{I}_2-\frac{1}{2}\mathbf{E}_t^{(i)}\bigg)\dot{\mathbf{x}}_t \\
&\phantom{= C_\perp^{(i)}(s)\bigg[}+\frac{L_0^2}{4}\dot{\theta}_t \mathbf{n}_t^{(0)}\cdot\bigg(\mathbf{I}_2-\frac{1}{2}\mathbf{E}_t^{(i)}\bigg)\mathbf{n}_t^{(0)} 
    + \sigma_*^{(i)}\frac{L_0}{2}\dot{\theta}_t s\, \mathbf{n}_t^{(i)}\cdot\bigg(\mathbf{I}_2-\frac{1}{2}\mathbf{E}_t^{(i)}\bigg)\mathbf{n}_t^{(0)}  \\
&\phantom{= C_\perp^{(i)}(s)\bigg[}+\sigma_*^{(i)}\frac{L_0}{2} s\, \big(\dot{\theta}_t+i^2\dot{\alpha}_t^{(i)}\big) \mathbf{n}_t^{(0)}\cdot\mathbf{n}_t^{(i)} 
    +s^2\big(\dot{\theta}_t+i^2\dot{\alpha}_t^{(i)}\big) \bigg].
\end{split}
\end{equation}
Therefore, the total force and torque acting on each link are given by integrating \eqref{forcedensity} and \eqref{torquedensity}, respectively, over $s \in [0, L_i]$, and adding up all the contributions for $i=-1,0,1$.
Yet, since the expressions in \eqref{forcedensity} and \eqref{torquedensity} do not admit closed-form expressions when integrated with respect to~$s$, we proceed by using a  first-order linearization of the coefficients. This is done using the approximation $d_t^{(i)}(s) \simeq y_t$\,, which is valid when the swimmer is nearly parallel to the wall. 
We obtain the following expression
\begin{subequations}\label{136}
\begin{equation}\label{136_C_perp}
 C^{(i)}_{\perp}(s) \simeq 
 G^{(i)}_t + sK^{(i)}_t,
\end{equation}
where we have defined
\begin{eqnarray}
G^{(i)}_t &\coloneqq& 
2 \pi \mu \ell_t^{-1}\big[2 - \sigma_-^{(i)} L_0 \sin(\theta_t)\, y_t^{-1}\ell_t^{-1} \big]\label{G_t} \\ 
K^{(i)}_t &\coloneqq& 
-  4 \pi \mu \sigma_+^{(i)}\,
\sin \big(\theta _t + i^2 \alpha^{(i)}_t\big)\, y_t^{-1}\ell_t^{-2}, \label{K_t}
\end{eqnarray}
\end{subequations}
with $\ell_t\coloneqq \log (2 y_t/r)$.
Thanks to \eqref{136}, we can integrate \eqref{forcedensity} and \eqref{torquedensity} with respect to $s\in [0, L_i]$ (for $i=\pm1,0$) 
and obtain analytic expressions. 
The resulting total force and torque exerted by the fluid on the swimmer are then given by
\begin{subequations}\label{total_force_torque}
\begin{eqnarray}
\mathbf{F}_t & = & \sum _{i = -1}^1 \int_0^{L_i} \mathbf{f}_t^{(i)}(s) \,\mathrm{d} s
= \mathbf{A}_t \dot{\mathbf{x}}_t + \mathbf{B}_t \dot{\boldsymbol{\alpha}}_t + \boldsymbol{\gamma}_t \dot{\theta}_t\,, \label{total_force} \\
T_t & = & \sum _{i = -1}^1 \int_0^{L_i} \tau_t^{(i)}(s) \,\mathrm{d}s 
=  \boldsymbol{\gamma}_t \cdot \dot{\mathbf{x}}_t + \mathbf{b}_t \cdot \dot{\boldsymbol{\alpha}}_t + \delta_t\dot{\theta}_t\,. \label{total_torque}
\end{eqnarray}
\end{subequations}
In \eqref{total_force_torque}, we have defined the $2\times2$ matrices $\mathbf{A}_t$ and $\mathbf{B}_t$ by
\begin{eqnarray*}
\mathbf{A}_t & \coloneqq & 
    \sum_{i=-1}^1 \bigg(G_t^{(i)}L_i+K_t^{(i)}\frac{L_i^2}{2}\bigg)\bigg(\mathbf{I}_2-\frac{1}{2}\mathbf{E}_t^{(i)}\bigg), \\
\mathbf{B}_t & \coloneqq & \bigg[-\bigg(G_t^{(-1)}\frac{L_{-1}^2}{2}+K_t^{(-1)}\frac{L_{-1}^3}{3}\bigg)\mathbf{n}_t^{(-1)} \bigg| \bigg(G_t^{(1)}\frac{L_{1}^2}{2}+K_t^{(1)}\frac{L_{1}^3}{3}\bigg)\mathbf{n}_t^{(1)}\bigg]
\end{eqnarray*}
(notice that $\mathbf{A}_t$ is symmetric), the vectors $\boldsymbol{\gamma}_t$ and $\mathbf{b}_t$ by
\begin{eqnarray*}
\boldsymbol{\gamma}_t  &\coloneqq& \sum_{i=-1}^1 \bigg[\bigg(G_t^{(i)}\frac{L_i^2}{2}+K_t^{(i)}\frac{L_i^3}{3}\bigg)\sigma_+^{(i)}\mathbf{n}_t^{(i)} 
+\frac{L_0}{2}\sigma_-^{(i)}\bigg(G_t^{(i)}L_i+K_t^{(i)}\frac{L_i^2}{2}\bigg) \bigg(\mathbf{I}_2-\frac{1}{2}\mathbf{E}_t^{(i)}\bigg) \mathbf{n}_t^{(0)}\bigg],\quad \\
\mathbf{b}_t &\coloneqq& \Bigg[ \frac{L_0}{2}\bigg(G_t^{(-1)}\frac{L_{-1}^2}{2} + K_t^{(-1)} \frac{L_{-1}^3}{3} \bigg)\mathbf{n}_t^{(0)} \cdot \mathbf{n}_t^{(-1)} + \bigg(G_t^{(-1)}\frac{L_{-1}^3}{3} + K_t^{(-1)} \frac{L_{-1}^4}{4} \bigg)\bigg| \\
    && \phantom{\bigg[} \frac{L_0}{2}\bigg(G_t^{(1)}\frac{L_{1}^2}{2} + K_t^{(1)} \frac{L_{1}^3}{3} \bigg)\mathbf{n}_t^{(0)} \cdot \mathbf{n}_t^{(1)} + \bigg(G_t^{(1)}\frac{L_{1}^3}{3} + K_t^{(1)} \frac{L_{1}^4}{4} \bigg)\bigg] ,
\end{eqnarray*}
and the scalar $\delta_t$ by
\begin{equation*}
\begin{split}
\delta_t \coloneqq & \sum_{i = -1}^1 \Bigg[ \frac{L_0^2}{4}\bigg(G_t^{(i)}L_i + K_t^{(i)} \frac{L_i^2}{2}\bigg) \mathbf{n}_t^{(0)}\cdot \bigg(\mathbf{I}_2-\frac{1}{2}\mathbf{E}_t^{(i)}\bigg)\mathbf{n}_t^{(0)}  \\
&\phantom{\sum_{i = -1}^1 \Bigg[} + \frac{L_0}{2}\sigma_*^{(i)} \bigg(G_t^{(i)}\frac{L_i^2}{2} + K_t^{(i)} \frac{L_i^3}{3} \bigg) \mathbf{n}_t^{(0)}\cdot \bigg(2\mathbf{I}_2-\frac{1}{2}\mathbf{E}_t^{(i)}\bigg)\mathbf{n}_t^{(0)} +  G_t^{(i)}\frac{L_i^3}{3} + K_t^{(i)} \frac{L_i^4}{4} \bigg].
\end{split}
\end{equation*}

Finally, from \eqref{total_force_torque}, the equations of motion can be written compactly as
\begin{equation} \label{motion}
\!\!\!\!\! \mathbf{0}=\begin{bmatrix}
\mathbf{F}_t \\
T_t
\end{bmatrix}
=
\begin{bmatrix}
\mathbf{A}_t & \boldsymbol{\gamma}_t & \mathbf{B}_t \\
\boldsymbol{\gamma}_t^\top & \delta_t & \mathbf{b}_t^\top
\end{bmatrix}
\begin{bmatrix}
\dot{\mathbf{x}}_t \\
\dot{\theta}_t \\
\dot{\boldsymbol{\alpha}}_t
\end{bmatrix}
=
\begin{bmatrix}
    (\mathbf{A}_t)_{11} & (\mathbf{A}_t)_{12} & (\boldsymbol{\gamma}_t)_1 & (\mathbf{B}_t)_{11} & (\mathbf{B}_t)_{12} \\
    (\mathbf{A}_t)_{21} & (\mathbf{A}_t)_{22} & (\boldsymbol{\gamma}_t)_2 & (\mathbf{B}_t)_{21} & (\mathbf{B}_t)_{22} \\
    (\boldsymbol{\gamma}_t)_1 & (\boldsymbol{\gamma}_t)_2 & \delta_t & (\mathbf{b}_t)_1 & (\mathbf{b}_t)_2 
\end{bmatrix}
\begin{bmatrix}
\dot{x}_t \\
\dot{y}_t \\
\dot{\theta}_t \\
\dot{\alpha}_t^{(-1)} \\
\dot{\alpha}_t^{(1)} 
\end{bmatrix}.
\end{equation}

\section{Controllability}
To study the controllability of the system, we employ tools from Geometric Control Theory. 
Starting from equation \eqref{motion}, we introduce two control functions $t\mapsto u_t^{(\pm1)}$ corresponding to the shape changes governed by $\dot{\alpha}_t^{(\pm1)}$, respectively. 
Eventually, we will define the associated control fields $\mathbf{h}_t^{(\pm1)}$ which will be used to compute Lie brackets and analyze the controllability of the system.

By introducing
\[
\mathbf{M}_t \coloneqq 
\begin{bmatrix}
\mathbf{A}_t  & \boldsymbol{\gamma}_t \\
\boldsymbol{\gamma}_t^\top  & \delta_t 
\end{bmatrix},\quad
\mathbf{g}^{(-1)}_t \coloneqq 
\begin{bmatrix}
    (\mathbf{B}_t)_{11} \\
    (\mathbf{B}_t)_{21} \\
    (\mathbf{b}_t)_1
\end{bmatrix},
\quad\text{and}\quad
\mathbf{g}^{(1)}_t \coloneqq 
\begin{bmatrix}
    (\mathbf{B}_t)_{12} \\
    (\mathbf{B}_t)_{22} \\
    (\mathbf{b}_t)_2
\end{bmatrix},
\]
and the control inputs $[u_t^{(-1)}, u_t^{(1)}]^\top \coloneqq [\dot{\alpha}_t^{(-1)}, \dot{\alpha}_t^{(1)}]^\top$,

the equations of motion \eqref{motion} can be rewritten as

\begin{equation}\label{motion_GRM}
\begin{bmatrix}
\mathbf{M}_t& \mathbf{0}_{3 \times 2} \\
\mathbf{0}_{2 \times 3} & \mathbf{I}_{2}
\end{bmatrix}
\begin{bmatrix}
\dot{\mathbf{x}}_t \\
\dot{\theta}_t \\
\dot{\boldsymbol{\alpha}}_t
\end{bmatrix}
= -
\begin{bmatrix}
\mathbf{g}_t^{(-1)} & \mathbf{g}_t^{(1)} \\
1 & 0 \\
0 & 1
\end{bmatrix}
\begin{bmatrix}
u_t^{(-1)} \\
u_t^{(1)}
\end{bmatrix}.
\end{equation}
The matrix $\mathbf{M}_t$ is usually referred to as the \textit{grand resistance matrix}; it is symmetric and invertible around the ``flat'' ($\alpha_t^{(\pm1)}=0$) horizontal ($\theta_t=0$) configuration, at which we have $\det \mathbf{M}_t=8\pi^3\mu^3(L_{-1}+L_0+L_1)^5/3\ell_t^3\neq0$. 
By continuity of the determinant with respect to the matrix entries, $\det\mathbf{M}_t$ remains different from zero also at configurations $\theta_t\,,\alpha_t^{(\pm1)}\approx0$.
This allows us to solve \eqref{motion_GRM} for the position and shape variables and write it as a nonlinear driftless affine control system
\begin{equation}\label{control_sys}
\begin{bmatrix}
\dot{\mathbf{x}}_t \\
\dot{\theta}_t\\
\dot{\boldsymbol{\alpha}}_t
\end{bmatrix}
\;=\;
\mathbf{h}_t^{(-1)}u_t^{(-1)} + \mathbf{h}_t^{(1)}u_t^{(1)}
\end{equation}
where the control vector fields $\mathbf{h}_t^{(\pm1)}$ are defined by
\[
\mathbf{h}_t^{(-1)} \coloneqq
\begin{bmatrix}
-\mathbf{M}_t^{-1}\mathbf{g}^{(-1)}_t \\
1\\
0
\end{bmatrix},
\qquad
\mathbf{h}_t^{(1)} \coloneqq
\begin{bmatrix}
-\mathbf{M}_t^{-1}\mathbf{g}^{(1)}_t \\
0\\
1
\end{bmatrix}.
\]
We are now ready to investigate the local controllability of \eqref{control_sys}. 
To simplify computations, we suppose that the outer links have length $L_{-1}=L_1=L$ and that the central one has length $L_0=\lambda L$, for some $\lambda\in (0,+\infty)$.
\begin{theorem}\label{main_thm}
Let $L>0$ and let $L_{\pm1}=L$, $L_0=\lambda L$, for a certain $\lambda>0$. 
Then control system \eqref{control_sys} is locally controllable around an aligned configuration parallel to the wall, \textit{i.e.}, for $\theta_t\,,\alpha_t^{(\pm1)}\approx0$.
\end{theorem}
\begin{proof}
In order to prove the local controllability of the system \eqref{control_sys} we resort to the Chow--Rashewskii theorem \cite[Theorem~5.9]{Agrachev}, which states that a sufficient condition for local controllability at a certain configuration is that the Lie algebra generated by the control vector fields at that configuration has dimension equal to the dimension of the tangent space of the overall system. 
In our specific case, we have to show that the control vector fields $\mathbf{h}_t^{(\pm1)}$ and their iterated Lie brackets at $(x_t,y_t,\theta_t , \alpha_t^{(-1)} , \alpha_t^{(1)} )=(x_t,y_t,0,0,0)\eqqcolon \mathbf{s}_t^0$ generate a space of dimension five. 
Through explicit computations, we obtain
\begin{equation*}
\mathbf{h}_t^1(\lambda)\coloneqq \mathbf{h}_t^{(-1)} \big|_{\mathbf{s}_t^0} =
\begin{bmatrix}
0 \\
\displaystyle\frac{L}{ 2 \lambda + 4}\\
\displaystyle\frac{ - 3 \lambda -4}{(\lambda + 2)^3}\\
1 \\
0
\end{bmatrix},\qquad
\mathbf{h}_t^2(\lambda)\coloneqq \mathbf{h}_t^{(1)} \big|_{\mathbf{s}_t^0} =
\begin{bmatrix}
0 \\
\displaystyle-\frac{L}{2 \lambda + 4}\\
\displaystyle\frac{- 3 \lambda -4}{(\lambda + 2)^3}\\
0 \\
1
\end{bmatrix},
\end{equation*}
\begin{equation}\label{net_displ_x}
\mathbf{h}_t^3(\lambda)\coloneqq [\mathbf{h}_t^{(-1)}, \mathbf{h}_t^{(1)}] \big|_{\mathbf{s}_t^0} =
\begin{bmatrix}
\displaystyle\frac{L \lambda (2 \lambda +3)}{(\lambda + 2)^4}\\
0 \\
0 \\
0 \\
0
\end{bmatrix},
\end{equation}
\begin{equation*}
\begin{aligned}
&\mathbf{h}_t^{4}(\lambda)\coloneqq[\mathbf{h}_t^{(1)},[\mathbf{h}_t^{(-1)}, \mathbf{h}_t^{(1)}]] \big|_{\mathbf{s}_t^0} =\\ &
\begin{bmatrix}
\displaystyle-\frac{L^2 \lambda (4 \lambda^4 +7\lambda^3-18\lambda^2-52\lambda-34 )}{6 (\lambda + 2)^7 \ell_t y_t}
  \\[3mm]
\displaystyle\frac{L\lambda^2\big(\ell_t^2 y_t^2(18\lambda^3+120\lambda^2+228\lambda+132)
- L^2 (2\lambda+3) (\lambda+2)^2 (2 + 
      \ell_t)\big)
   }{12 \ell_t^2 y_t^2(\lambda + 2)^7}\\[3mm]
\displaystyle\frac{\lambda \big(3 \ell_t^2 y_t^2(\lambda + 2)^3 - 
   L^2 ( \lambda+1)^3 (2 + 
      \ell_t)\big)}{
   \ell_t^2 y_t^2(\lambda + 2)^{7}}\\
0\\
0
\end{bmatrix},\\& 
\mathbf{h}_t^{5}(\lambda)\coloneqq[\mathbf{h}_t^{(-1)},[\mathbf{h}_t^{(-1)}, \mathbf{h}_t^{(1)}]] \big|_{\mathbf{s}_t^0} =\\ &
\begin{bmatrix}
\displaystyle\frac{L^2 \lambda (4 \lambda^4 +7\lambda^3-18\lambda^2-52\lambda-34 )}{6 (\lambda + 2)^7 \ell_t y_t}\\[3mm]
\displaystyle\frac{L\lambda^2\big(\ell_t^2 y_t^2(18\lambda^3+120\lambda^2+228\lambda+132)
- L^2 (2\lambda+3) (\lambda+2)^2 (2 + 
      \ell_t)\big)
   }{12 \ell_t^2 y_t^2(\lambda + 2)^7}\\[3mm]
\displaystyle - \frac{\lambda \big(3 \ell_t^2 y_t^2(\lambda + 2)^3 - 
   L^2 (\lambda+1)^3 (2 + 
      \ell_t)\big)}{
   \ell_t^2 y_t^2(\lambda + 2)^{7}}\\
0\\
0
\end{bmatrix}.
\end{aligned}
\end{equation*}

Finally, we can compute the determinant of the resulting $5 \times 5$ matrix, which reads

\[
D(\lambda) \coloneqq \mathrm{det}\big(\mathbf{h}_t^{1}(\lambda)\,|\,\mathbf{h}_t^{2}(\lambda) \,|\, \mathbf{h}_t^{3}(\lambda) \,|\, \mathbf{h}_t^{4}(\lambda) \,|\, \mathbf{h}_t^{5}(\lambda)\big)
=\frac{p(\lambda,y_t)}{6 (\lambda + 2)^{10} \ell_t^4 y_t^4}.
\]
where
$$
\begin{aligned}
p(\lambda,y_t)\coloneqq & L^2 \lambda^4(2\lambda+3)\big(L^4(\lambda+1)^3(\lambda+2)^2(2\lambda+3)(2+\ell_t)^2\\
&-3L^2\ell_t^2y_t^2(2+\ell_t)(8\lambda^6+81\lambda^5+324\lambda^4+678\lambda^3+800\lambda^2+512\lambda+140)\\
&+18\ell_t^4y_t^4(\lambda+2)^3(3\lambda^3+20\lambda^2+38\lambda+22)\big).
\end{aligned}$$
Since we are assuming $y_t \gg r$, we have $\ell_t>0$ so that the sign of $p(\lambda,y_t)$ is positive, since it is dominated by the term that is quartic in $\ell_t$\,. 
This implies that the set of vector fields
\[
\big\{\mathbf{h}_t^{1}(\lambda),\mathbf{h}_t^{2}(\lambda) , \mathbf{h}_t^{3}(\lambda) , \mathbf{h}_t^{4}(\lambda) , \mathbf{h}_t^{5}(\lambda)\big\}
\]
spans the 
tangent space at the reference configuration.
This proves the local controllability of the system around the configuration $\mathbf{s}_t^0$ for every $\lambda>0$.
\end{proof}
\begin{remark}
It is well known (see, e.g., \cite{Coron}) that the displacement of the system which starts in the aligned configuration parallel to the wall after a periodic control stroke with piecewise constant controls is given by a multiple of the first order bracket $\mathbf{h}_{t}^3(\lambda)$. 
From \eqref{net_displ_x}, it emerges that this results in a displacement parallel to the wall, independently from the distance of the swimmer from the wall, for all values of $\lambda$. 
Notice that this independence is not linked to the symmetric design of the three-link swimmer: indeed, if the three links all have different lengths, the Lie bracket at the aligned configuration parallel to the wall is given by
$$
\mathbf{h}_t^{3}(L_{-1},L_0,L_1)=
\begin{bmatrix}
\displaystyle\frac{(L_0 L_1 L_{-1} (L_1 (L_0 + L_1) + (L_0 + L_1) L_{-1} + L_{-1}^2))}{(L_0 + L_1 + L_{-1})^4}\\
0 \\
0 \\
0 \\
0
\end{bmatrix},
$$
and notice that $\mathbf{h}_t^3(L,\lambda L,L)=\mathbf{h}_t^3(\lambda)$.
Therefore, we can conclude that if the swimmer is parallel to the wall, after making a small periodic control stroke with piecewise controls, it will move parallel to the wall, regardless of the length of the links.
\end{remark}

\subsection{Swimmer tilted with respect to the wall}
We report here the displacement of the swimmer after a small periodic stroke, when it starts ``flat'', i.e., $\alpha_0^{(\pm1)}=0$, and tilted with respect to the wall at an angle $\theta_0$\,. The resulting displacement can be computed using the first order Lie bracket, evaluated at the configuration $\mathbf{s}^{\theta_0}\coloneqq(x_0\,,y_0\,,\theta_0,0,0)$, 
\begin{equation*}
\overline{\mathbf{h}}\coloneqq [\mathbf{h}_t^{(-1)}, \mathbf{h}_t^{(1)}] \big|_{\mathbf{s}^{\theta_0}}=
  \begin{pmatrix}
      v\cos(\theta_0) \\
     v\sin(\theta_0) \\
0\\
 0\\0
  \end{pmatrix},
\end{equation*}
where the norm of translational displacement is
$$
v=\|\overline{\mathbf{h}}\|=\displaystyle\frac{L \lambda\big(L^2 (
       3 \lambda+4) (2\lambda^2+5\lambda+4) \sin^2(\theta_0) - 
    12 ( 2 \lambda+3) \ell_0^2y_0^2\big)}{
 L^2 (\lambda+2)^4 \sin^2(\theta_0)(\lambda+2)^2 - 
  12 \ell^2_0y_0^2}\,.
$$
From these formulas it is evident that the resulting displacement is a translation in the direction $\mathbf{e}_0^{(0)}$, while the swimmer does not rotate. This means that the presence of the wall does not affect the direction of motion, but only its modulus. This behaviour contrasts with outcomes of numerical simulations of swimming spermatozoa and bacteria that more accurately account for swimming gaits and hydrodynamic interactions with the wall. In particular, such swimmers initially tilted towards a wall typically turn to become more aligned with the wall, ultimately swimming parallel to the wall or, in some cases, escaping from the wall \cite{shum_modelling_2010,smith_human_2009}. In experiments, swimming bacteria and sperm are also commonly observed to accumulate at surfaces \cite{berke_hydrodynamic_2008, frymier_three-dimensional_1995, rothschild_non-random_1963}, a characteristic that can be explained by hydrodynamic interactions \cite{LAUGA2006400}. This result is also in accordance with \cite{OR10} where a comparable behavior is seen, though not within a control setting. The fact that our current analysis does not  predict angular displacements as the swimmer approaches a wall could be due to the assumed form of the wall-dependent drag coefficients \eqref{coefficients}, which is a simplification.

Note that in the regime $y_0\gg r$ the quantity $v$ above is positive. Moreover for $y_0\gg r$ we have that $1/\ell_0^2$ is small, therefore we can expand $v$ in powers of $1/\ell_0^2$ obtaining 
$$
v=\frac{\lambda(3L(2\lambda+3)y_0^2-\frac{1}{\ell_0^2}L^3\sin^2(\theta_0)(\lambda+1)^3)}{3(\lambda+2)^4y_0^2}+o(1/\ell_0^2)\,.
$$

Thus at leading order the displacement is maximum in $\theta_0=0$, namely in the configuration parallel to the wall. This differs from the behavior of the swimmer in an unbounded fluid, for which the displacement is in the direction of $\mathbf{e}_0^{(0)}$ but has a magnitude that is independent of orientation. The presence of the wall thus makes the norm of the displacement vary with the angle of inclination.

\begin{remark}
    Notice that if $|\theta_0|$ is sufficiently small, the system remains controllable also at the tilted configuration. Indeed, since the vector fields and their Lie brackets are real analytic, their determinant is a continuous function. Therefore, the determinant in the configuration $(x_0,y_0,\theta_0,0,0)$ will remain different from zero, 
    proving the local controllability.
\end{remark}

\section{Numerical simulations}
We solve the control system \eqref{control_sys} in MATLAB using the \texttt{ode45} function, which is an implementation of the Runge--Kutta--Fehlberg 4(5) numerical ODE solver. For the numerical solution, we integrate the wall-dependent terms in the resistance matrix directly without the linearizing step \eqref{136_C_perp}. The controls $u_t^{(-1)}$ and  $u_t^{(1)}$ are defined as piecewise constant functions with an amplitude parameter $\xi$ such that the net displacement after one period of the control corresponds to the defined vector fields $\mathbf{h}^1_t, \ldots, \mathbf{h}^5_t$, respectively, in the limit of \(\xi \to 0\). 

The results we present focus on the vector field $\mathbf{h}^3_t=[\mathbf{h}^1_t,\mathbf{h}^2_t]$ since this corresponds to swimming along the direction of the links, starting from the straight configuration $\alpha^{(-1)}_0=\alpha^{(1)}_0=0$. We simulate trajectories using the 4-periodic controls
\begin{equation}
u_t^{(-1)} = \begin{cases} 
\xi, & 0 \leq t \;\mathrm{mod}\; 4 < 1,\\
0, & 1 \leq t \;\mathrm{mod}\; 4 < 2,\\
-\xi, & 2 \leq t \;\mathrm{mod}\; 4 < 3,\\
0, & 3 \leq t \;\mathrm{mod}\; 4 < 4,
\end{cases}\hspace{1cm}
u_t^{(1)} = \begin{cases} 
0, & 0 \leq t \;\mathrm{mod}\; 4 < 1,\\
\xi, & 1 \leq t \;\mathrm{mod}\; 4 < 2,\\
0, & 2 \leq t \;\mathrm{mod}\; 4 < 3,\\
-\xi, & 3 \leq t \;\mathrm{mod}\; 4 < 4.
\end{cases}
\end{equation}
The normalized displacements per cycle are given by
\begin{eqnarray}
    \boldsymbol\Delta = \begin{bmatrix}
        \Delta x \\
        \Delta y \\
        \Delta \theta \\
        \Delta \alpha^{(-1)}\\
        \Delta \alpha^{(1)}
    \end{bmatrix}&:=& \frac{1}{\xi^2}
    \begin{bmatrix}
    x_4 - x_0\\
    y_4 - y_0\\
    \theta_4 - \theta_0\\
    \alpha^{(-1)}_4 - \alpha^{(-1)}_0\\
    \alpha^{(1)}_4 - \alpha^{(1)}_0
    \end{bmatrix} = \mathbf{h}^3_0 + \begin{bmatrix} O(\xi)\\ O(\xi) \\ O(\xi) \\ 0 \\ 0\end{bmatrix} 
\end{eqnarray}
as $\xi \to 0$. In Fig.~\ref{fig:displacements_h3}, the $x$, $y$, and $\theta$ components of \(\boldsymbol\Delta\) are plotted along with their corresponding components of $\mathbf{h}^3_t$ starting from a straight configuration of the links parallel to the wall. The numerical results converge to the asymptotic result as \(\xi\to 0\). We note that for moderate values of $\xi\sim 0.1$, which may be used in practice, there is a drift away from the wall ($\Delta y > 0$) and simultaneous rotation toward the wall ($\Delta\theta < 0)$ when the swimmer is initially parallel to the wall.

\begin{figure}[h]
    \centering
    \includegraphics[width=0.5\linewidth]{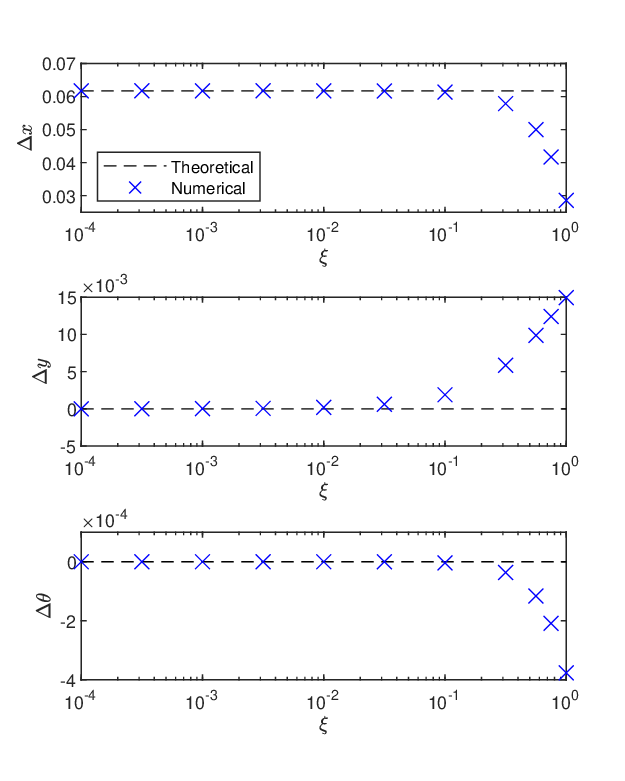}
    \caption{Numerically computed displacements $\Delta x$, $\Delta y$, and $\Delta\theta$ as functions of $\xi$ and corresponding theoretical values from the vector field $\mathbf{h}^3_t$. The initial state is $(x_0,y_0,\theta_0,\alpha^{(-1)}_0,\alpha^{(1)}_0) = (0,2,0,0,0)$. The geometrical parameters used are $L=1$, $\lambda = 1$, $r=0.01$.}
    \label{fig:displacements_h3}
\end{figure}

Additional effects of finite $\xi$ can be noted in Fig. \ref{fig:displacements_h3_tilt}, which shows the displacements as functions of the tilt angle $\theta_0$. The larger value $\xi=0.1$ results in smaller displacements $||\Delta\mathbf{x}||$ and larger deviations in $\Delta\theta$ away from the theoretical value $\Delta\theta = 0$ expected in the limit $\xi\to 0$. When the swimmer is nearly parallel with the wall, $\Delta\theta>0$ for $\theta_0>0$ and $\Delta\theta<0$ for $\theta_0<0$. 

\begin{figure}[h]
    \centering
    \includegraphics[width=0.5\linewidth]{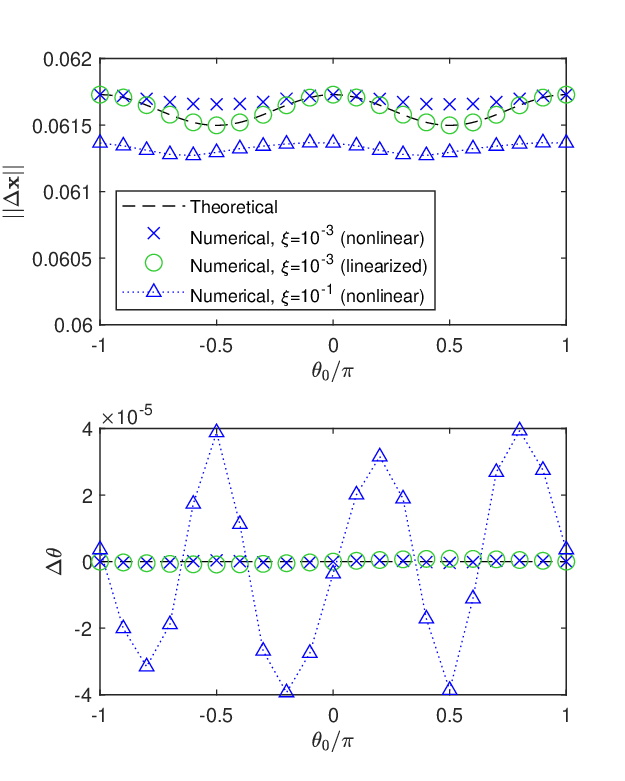}
    \caption{Comparisons between numerically computed and theoretical displacements $||\Delta \mathbf{x}|| = \sqrt{\Delta x^2 + \Delta y^2}$ and $\Delta\theta$ as functions of the tilt angle $\theta_0$ with respect to the wall for controls corresponding to the vector field $\mathbf{h}^3_t$. The initial state is $(x_0,y_0,\theta_0,\alpha^{(-1)}_0,\alpha^{(1)}_0) = (0,2,\theta_0,0,0)$. The geometrical parameters used are $L=1$, $\lambda = 1$, $r=0.01$. The nonlinear numerical solutions use the nonlinear drag coefficients \eqref{coefficients} whereas the linearized solution uses the approximation \eqref{136}.}
    \label{fig:displacements_h3_tilt}
\end{figure}

\begin{figure}[h]  
    \centering
    \includegraphics[width=0.5\linewidth]{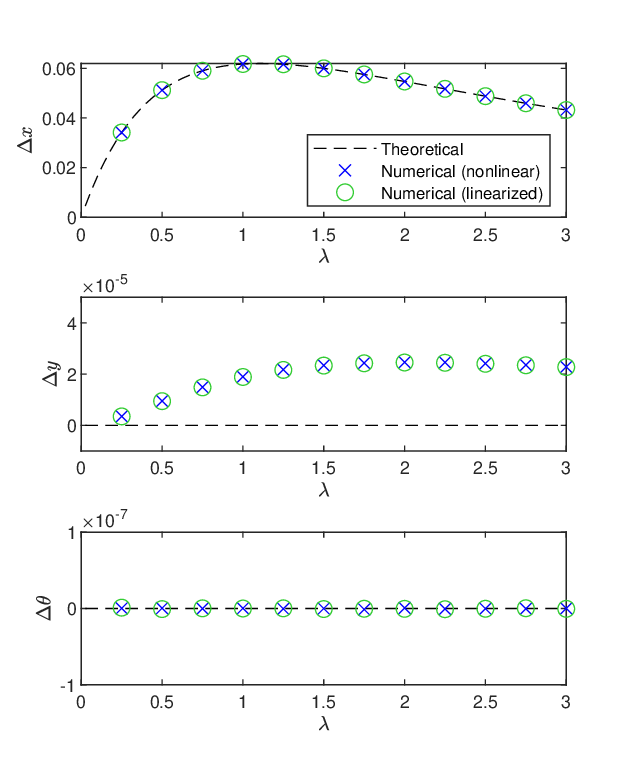}
    \caption{Comparisons between numerically computed and theoretical displacements $\Delta x$, $\Delta y$, and $\Delta\theta$ as functions of the link length ratio $\lambda$ for controls corresponding to the vector field $\mathbf{h}^3_t$ with amplitude $\xi=10^{-3}$. The initial state is $(x_0,y_0,\theta_0,\alpha^{(-1)}_0,\alpha^{(1)}_0) = (0,2,0,0,0)$. The common geometrical parameters used are $L=1$ and $r=0.01$.}
    \label{fig:displacements_h3_lambda}
\end{figure}

To verify the effect of modifying the ratio $\lambda$ between the length of the central link and that of the outer links, we numerically solve the control system with $\xi = 10^{-3}$. As shown in Fig.~\ref{fig:displacements_h3_lambda}, the numerical results agree closely with the theoretical prediction based on \eqref{net_displ_x}.
 Indeed, from the $x$-component of $\mathbf{h}_t^3(\lambda)$, it can be determined that the maximum displacement is obtained at $\lambda=(\sqrt{97}-1)/8\approx1.10611$.
 \clearpage
\section{Conclusions}
In this work we have investigated the dynamics and controllability of a Purcell three-link swimmer moving through a viscous fluid in a plane perpendicular to a rigid wall. Starting from a wall-corrected Resistive Force Theory model, we derived an explicit control system valid for configurations that are nearly parallel to the boundary. Using tools from Geometric Control Theory, we proved that the swimmer remains locally controllable around aligned configurations, showing that the wall does not hinder the ability of the system to generate arbitrary small rigid motions. On the contrary, the boundary breaks some of the symmetries of the unbounded case and enriches the structure of the Lie algebra generated by the control vector fields. The Lie bracket analysis further clarifies the geometric nature of the displacement induced by small periodic strokes even in a tilted configuration. All our theoretical predictions are in good agreement with numerical simulations.

Several directions for future research naturally emerge from our analysis.
In the two-dimensional setting, an improvement would consist in replacing the linearized drag approximation with a more accurate treatment of wall effects, for instance by retaining higher-order terms in expansion \eqref{136_C_perp}. This would allow one to quantify more precisely the range of validity of the present model and to explore regimes in which the swimmer approaches the wall more closely.
Allowing more general three-dimensional motion of the swimmer, several qualitatively new phenomena are expected. First, the configuration space becomes higher-dimensional, with additional orientation degrees of freedom, and the presence of a planar wall breaks rotational symmetry in a non-trivial way. 
A 3D Purcell-type swimmer near a no-slip boundary could exhibit out-of-plane reorientation, circular trajectories, or stable limit cycles analogous to those observed for helical bacteria near surfaces. From a control perspective, it would be natural to investigate whether controllability persists near configurations parallel to the wall and how the Lie algebra structure is modified by the additional rotational degrees of freedom.

\bibliographystyle{unsrt}
\bibliography{bibliografia}

\end{document}